\documentclass[conference]{IEEEtran}
\usepackage[pdftex]{graphicx}
\usepackage{float}
\usepackage{latexsym}
\usepackage[cmex10]{amsmath}
\usepackage{amssymb}
\usepackage{amsmath}
\usepackage{paralist}
\usepackage{leftidx}
\usepackage{mathrsfs}
\usepackage{multirow}
\usepackage{slashbox}
\usepackage{amsthm}
\usepackage{background}
\usepackage{algorithm}
\usepackage{algorithmic}
\usepackage[norule]{footmisc}
\IEEEoverridecommandlockouts

\theoremstyle{definition}
\newtheorem{defn}{Definition}
\theoremstyle{proposition}
\newtheorem{prop}{Proposition}
\theoremstyle{corollary}
\newtheorem{cor}{Corollary}
\theoremstyle{theorem}

\theoremstyle{lemma}
\newtheorem{lem}{Lemma}
\newtheorem*{remark}{Remark}
\begin{document}
\SetBgContents{}
\title{Mobility Increases the Data Offloading Ratio in D2D Caching Networks}
\author{\IEEEauthorblockN{Rui Wang$^*$, Jun Zhang$^*$, S.H. Song$^*$ and K. B. Letaief$^*$$^\dag$, \emph{Fellow, IEEE} }
	\IEEEauthorblockA{$^*$Dept. of ECE, The Hong Kong University of Science and Technology, $^\dag$Hamad Bin Khalifa University, Doha, Qatar\\
		Email: $^*$\{rwangae, eejzhang, eeshsong, eekhaled\}@ust.hk, $^\dag$kletaief@hbku.edu.qa} \thanks{This work was supported by the Hong Kong Research Grants Council under Grant No. 610113.}}
\maketitle
\begin{abstract}

Caching at mobile devices, accompanied by device-to-device (D2D) communications, is one promising technique to accommodate the exponentially increasing mobile data traffic. While most previous works ignored user mobility, there are some recent works taking it into account. However, the duration of user contact times has been ignored, making it difficult to explicitly characterize the effect of mobility. In this paper, we adopt the alternating renewal process to model the duration of both the contact and inter-contact times, and investigate how the caching performance is affected by mobility. The \emph{data offloading ratio}, i.e., the proportion of requested data that can be delivered via D2D links, is taken as the performance metric. We first approximate the distribution of the \emph{communication time} for a given user by beta distribution through moment matching. With this approximation, an accurate expression of the data offloading ratio is derived. For the homogeneous case where the average contact and inter-contact times of different user pairs are identical, we prove that the data offloading ratio increases with the user moving speed, assuming that the transmission rate remains the same. Simulation results are provided to show the accuracy of the approximate result, and also validate the effect of user mobility.
\end{abstract}
\IEEEpeerreviewmaketitle

\section{Introduction}

The mobile data traffic is growing at an exponential rate, among which mobile video accounts for more than a half \cite{forecast2016cisco}. Caching popular contents at helper nodes or user devices is a promising approach to reduce the data traffic on the backhaul links, as well as improving the user experience of video streaming applications  \cite{d2d-cache,jcache}. In comparison with the commonly considered femto-caching system
, caching at devices enjoys a unique advantage, i.e., the devices' aggregate caching capacity grows with the number of devices \cite{d2d-cache}. Moreover, device caching can promote device-to-device (D2D) communications, where nearby mobile devices may communicate directly rather than being forced to communicate through the base station (BS)\cite{design}. 

Recently, caching in D2D networks has attracted lots of attentions. In \cite{scaling}, the scaling behavior of the number of D2D collaborating links was identified. Three concentration regimes, classified by the concentration of the file popularity, were investigated. The outage-throughput tradeoff and optimal scaling laws of both the throughput and outage probability were studied in \cite{tradeoff}. Two coded caching schemes, i.e., centralized and decentralized, were proposed in \cite{fundamentallimits}, where the contents are delivered via broadcasting.

So far, an important characteristic of mobile users, i.e., the user mobility,  has been ignored in previous studies of D2D caching networks. There are some works starting to consider the effect of user mobility. Effective methodologies to utilize the user mobility information in caching design were discussed in \cite{magmobility}. 
In \cite{mobilitycodedcaching}, the effect of mobility was evaluated in D2D networks with coded caching, with the conclusion that mobility can improve the scaling law of throughput. This result was based on the assumption that the user locations are random and independent in each time slot, which failed to take into account the temporal correlation. 

The inter-contact model, which considers the temporal correlation of the user mobility,  has been widely applied \cite{exintercontactmodel}, where the timeline for an arbitrary pair of mobile users are divided into \emph{contact times} and \emph{inter-contact times}. Specifically, the \emph{contact times} denote the time intervals when the mobile users are located within the transmission range. Correspondingly, the \emph{inter-contact times} denote the time intervals between contact times \cite{pocket}. This model has been used to develop device caching schemes to exploit the user mobility pattern in \cite{mobilitycaching}. The throughput-delay scaling law was developed by characterizing the inter-contact pattern of the random walk model \cite{scalingmobility}. In these works, it was assumed that a fixed amount of data can be delivered within one contact time, while the duration of the contact times was not considered. However, as the user moving speed will affect the durations of both the contact and inter-contact times, it is critical to account for their effects when investigating the impact of user mobility on caching performance.

In this paper, we shall analytically evaluate the effect of mobility in D2D caching networks, by adopting an alternating renewal process to model the mobility pattern so that both the contact and inter-contact times are accounted for. The \emph{data offloading ratio}, which is defined as the proportion of data that can be obtained via D2D links, is adopted as the performance metric. The main contribution is an approximate expression for the data offloading ratio, for which the main difficulty is to deal with multiple alternating renewal processes. We tackle it by first deriving the expectation and variance of the \emph{communication time} of a given user, and then use a beta random variable to approximate it by moment matching. Furthermore, we investigate the effect of mobility in a homogeneous case, where the average contact and inter-contact times for all the user pairs are the same. In the low-to-medium mobility scenario, by assuming that the transmission rate is irrelevant to the user speed, it is proved that the data offloading ratio increases with the user speed for any caching strategy that does not cache the same contents at all devices. Simulation results validates the accuracy of the derived expression, as well as the effect of the user mobility.

\section{System Model and Performance Metric}

In this section, we will first introduce the alternating renewal process to model the user mobility pattern, and discuss the caching and file delivery models. Then, the performance metric, i.e., the data offloading ratio, will be defined.

\subsection{User Mobility Model}
\begin{figure}[!t]
  \centering
  \includegraphics[width=3in]{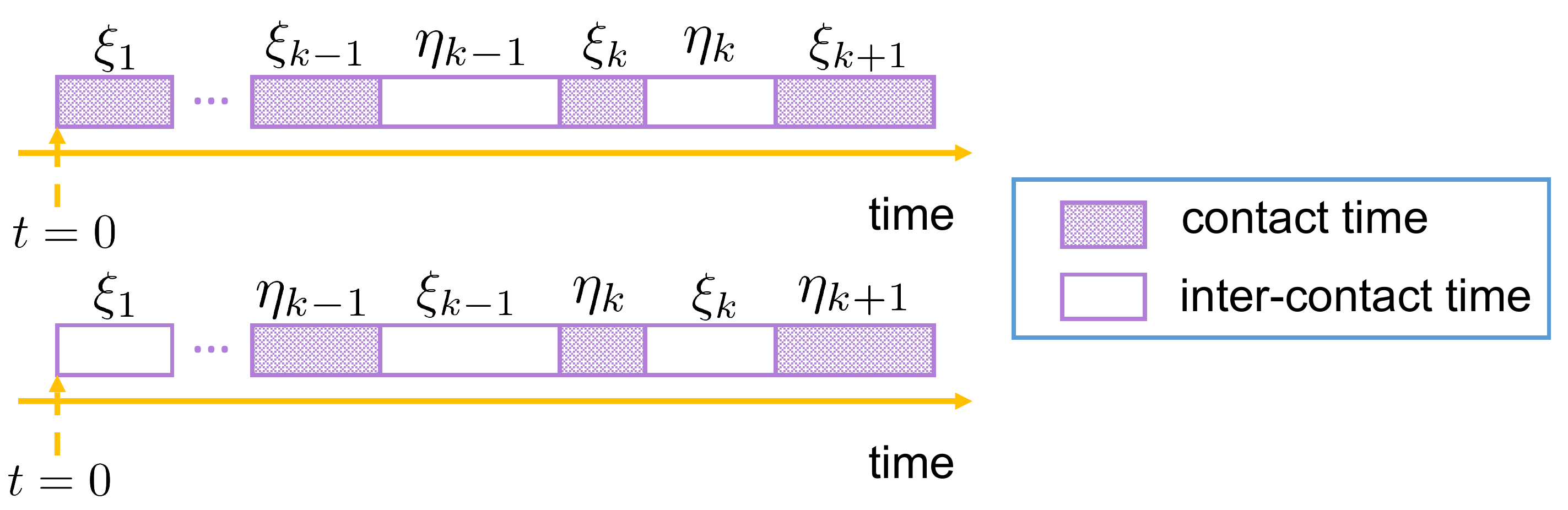}
  \caption{The timeline for an arbitrary pair of mobile users.}
  \label{intercontact}
\end{figure}

The inter-contact model, which captures the temporal correlation of the user mobility \cite{exintercontactmodel}, is used to model the user mobility pattern. Specifically, the timeline of each pair of users is divided into  \emph{contact times}, i.e, the times when the users are in the transmission range, and \emph{inter-contact times}, i.e., the times between consecutive contact times. Considering that contact times and inter-contact times appear alternatively in the timeline of a pair of users, similar to \cite{renewalmodel}, an alternating renewal process is applied to model the pairwise contact pattern, as defined below \cite{renewalprocess}.

\begin{defn}
Consider a stochastic process with state space $\{A,B\}$, and the successive durations for the system to be in states $A$ and $B$ are denoted as $\xi_k,k=1,2,\cdots$ and $\eta_k,k=1,2,\cdots$, respectively, which are i.i.d.. Specifically, the system starts at state $A$ and remains for $\xi_1$, then switches to state $B$ for $\eta_1$, then backs to state $A$ for $\xi_2$, and so forth. 
Let $\psi_k=\xi_k+\eta_k$. The counting process of $\psi_k$ is called as an \emph{alternating renewal process}.
\end{defn}
As shown in Fig. \ref{intercontact}, if the pair of users is in contact at $t=0$, $\xi_k$ and $\eta_k$ represent the contact times and inter-contact times, respectively; otherwise, $\xi_k$ and $\eta_k$ represent the inter-contact times and contact times, respectively. It was shown in \cite{exintercontact} that exponential curves well fit the distribution of inter-contact times, while in \cite{excontact}, it was identified that exponential distribution is a good approximation for the distribution of the contact times. Thus, same as \cite{renewalmodel}, we assume that the contact times and inter-contact times follows independent exponential distributions. For simplicity, the timelines of different user pairs are assumed to be independent. Specifically, we consider $N_u$ users in a network, and the index set of the users is denoted as $\mathcal{S}=\{1,2,\cdots,N_u\}$. The contact times and inter-contact times of users $i  \in \mathcal{S}$ and $j \in \mathcal{S} \backslash \{i\}$ follow independent exponential distributions with parameters $\lambda^C_{i,j}$ and $\lambda^I_{i,j}$, respectively.
\subsection{Caching and File Transmission Model}
\begin{figure}[!t]
	\centering
	\includegraphics[width=2.6in]{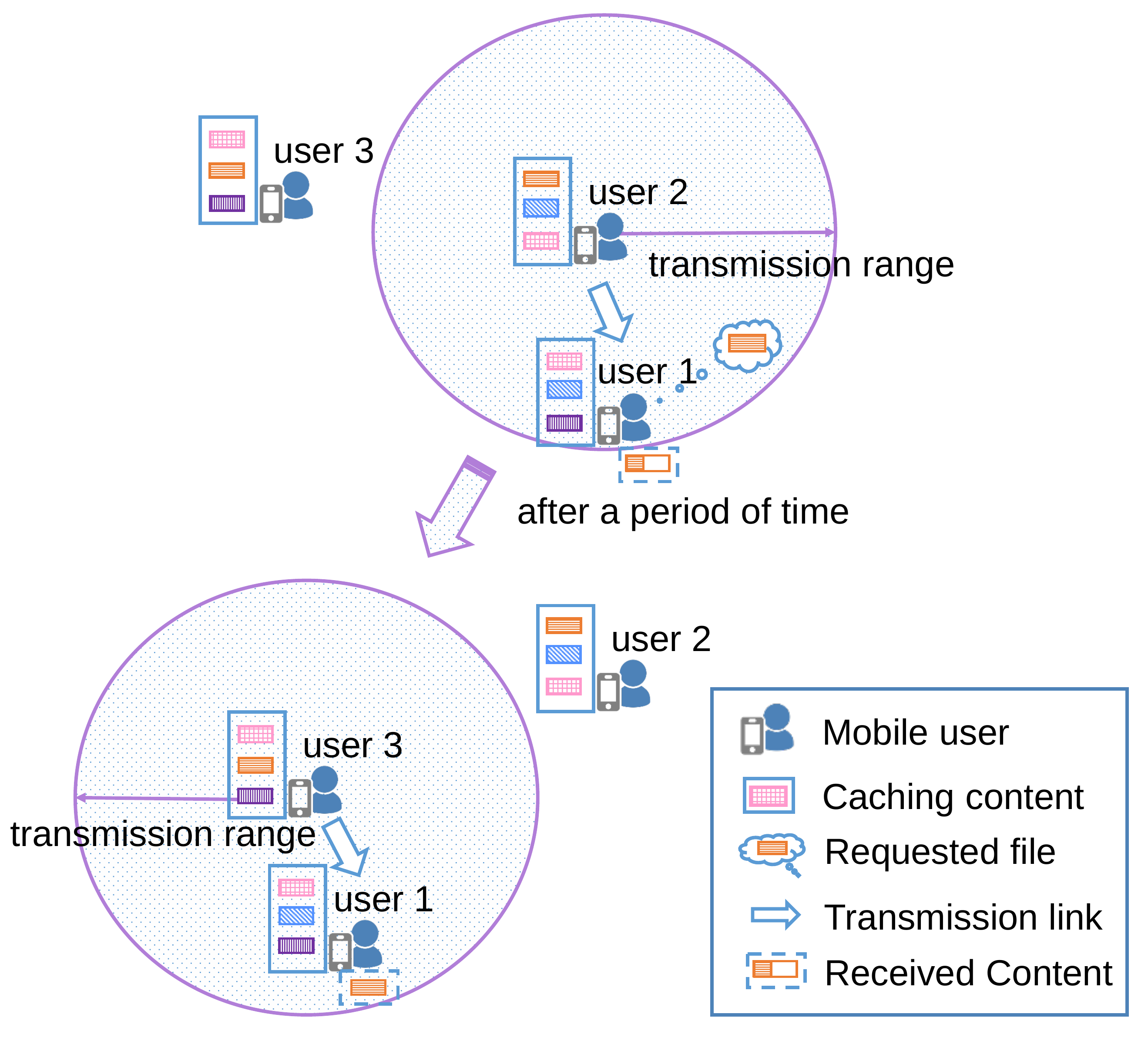}
	\caption{A sample network with three mobile users.}
	\label{model}
\end{figure}

There is a library with $N_f$ files, whose index set is denoted as $\mathcal{F}=\{1,2,\cdots,N_f\}$, each with size $C$. Each user device has a limited storage capacity, and each file can be completely cached or not cached at all at each user device. Specifically, the caching placement is denoted as
\begin{equation}
x_{j,f}=
\begin{cases}
1, \text{if user $j$ caches file $f$}, \\
0, \text{if user $j$ does not cache file $f$},
\end{cases}
\end{equation} 
where $j \in \mathcal{S}$ and $f \in \mathcal{F}$. User $i \in \mathcal{S}$, is assumed to request a file $f \in \mathcal{F}$ with probability $p^r_{i,f}$, where $\sum \limits_{f \in \mathcal{F} } p^r_{i,f} =1$. When a user requests a file $f$, it will first check its own cache, and then download the file from the users that are in contact and store file $f$, with a fixed transmission rate, denoted as $R$. If the user cannot get the whole file within a certain delay threshold, denoted as $T^d$, it will download the remaining part from the BS. We assume that the delay threshold is larger than the time required to download each content, i.e., $T^d>\frac{C}{R}$. Fig. \ref{model} shows a sample network, where user $1$ gets part of the requested file during the contact time with user $2$, then gets the whole file after the contact time with user $3$. 

\subsection{Performance metric}
\begin{figure}[!t]
	\centering
	\includegraphics[width=3.5in]{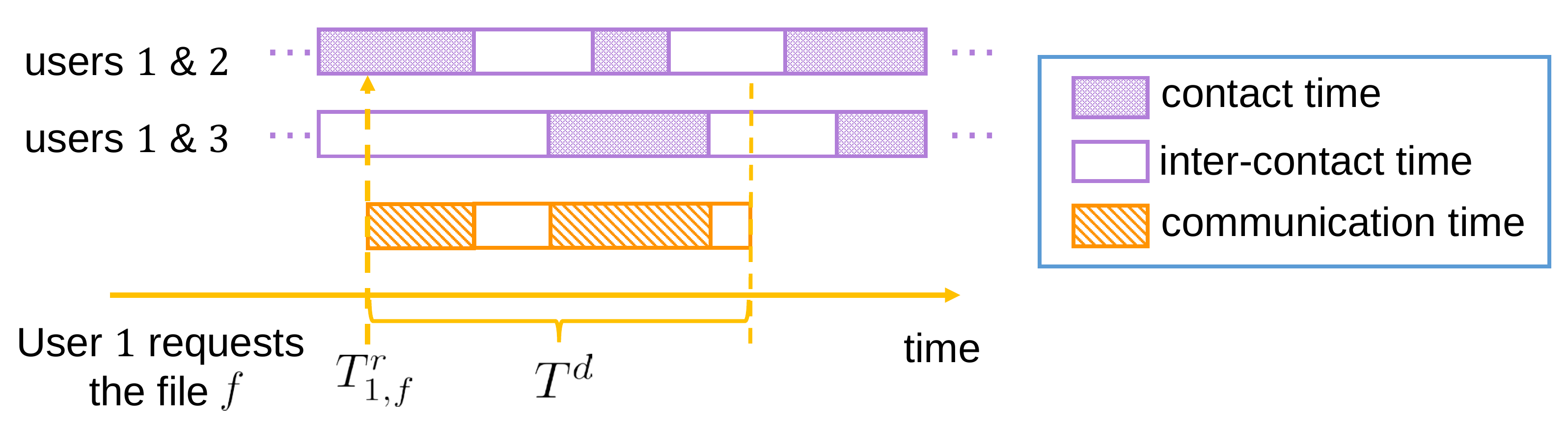}
	\caption{An illustration of the communication time.}
	\label{transmissiontime}
\end{figure}

The \emph{data offloading ratio}, which is defined as the percentage of requested content that can be obtained via D2D links rather than downloading from the BS, is used as the performance metric. Specifically, the data offloading ratio for user $i \in \mathcal{S}$ is defined as
\begin{equation}
\mathcal{P}_i=\sum \limits_{f \in \mathcal{F}} p^r_{i,f} \left\{ x_{i,f} + \frac{  (1-x_{i,f})\mathbb{E} _{D_{i,f}}\left[ \min \left( D_{i,f} ,C \right) \right]}{C} \right\},
\end{equation}
where $D_{i,f}$ denotes the amount of requested data that can be delivered via D2D links when user $i$ requests file $f$. Since a fixed transmission rate is assumed, $D_{i,f}$ can be written as $D_{i,f}=R T^c_{i,f}$, where $T^c_{i,f}$ is the \emph{communication time} for user $i$ to download file $f$ from other users caching file $f$ within time $T^d$. We assume that user $i$ can download file $f$ while at least one user caching file $f$ is in contact, where the handover time is ignored. Fig. \ref{transmissiontime} shows the communication time of user $1$ in Fig. \ref{model}. Then, the average data offloading ratio is
\begin{align} \label{define_ratio}
&\mathcal{P}= \notag  \\
&\frac{1}{N_u}\sum \limits_{i \in \mathcal{S}} \sum \limits_{f \in \mathcal{F}} p^r_{i,f} \left\{ x_{i,f} + \frac{ (1-x_{i,f}) \mathbb{E}_{T^c_{i,f} } \left[ \min \left( R T^c_{i,f} ,C \right) \right]}{C} \right\}.
\end{align}

In the following, we will evaluate the data offloading ratio given in (\ref{define_ratio}) for any given caching strategy, and investigate the effect of user mobility on caching performance. 

\section{Data Offloading Ratio Analysis}
The main difficulty of evaluating the data offloading ratio is to find the distribution of the communication time. As this distribution is highly complicated, instead of deriving it directly, we will develop an accurate approximation. In this section, we will first approximate the distribution of the communication time by a beta distribution, and then, an approximation of the data offloading ratio will be obtained.

\subsection{Communication time analysis}
To help analyze the communication time, we first define some stochastic processes.
\begin{defn}
Define $\mathbf{H}_{i,j}$, where $i \in \mathcal{S}$ and $j \in \mathcal{S} \backslash \{i\}$, as the continuous-time random process, i.e., $\mathbf{H}_{i,j}=\{H_{i,j}(t), t \in (0,\infty)\}$ with state space $\{0,1\}$, where $H_{i,j}(t)=1$ means that users $i$ and $j$ are in contact at the time instant $t$; otherwise $H_{i,j}(t)=0$. The durations of staying in states $1$ and $0$ follow i.i.d. exponential distributions with parameter $\lambda^C_{i,j}$ and $\lambda^I_{i,j}$, respectively.  
\end{defn}
\begin{defn}
Define $\mathbf{H}_i^f$, where $i \in \mathcal{S}$ and $f \in \mathcal{F}$, as the continuous-time random process, i.e., $\mathbf{H}^f_{i}=\{H^f_{i}(t), t \in (0,\infty)\}$ with state space $\{0,1\}$, where $H^f_{i}(t)=1$ means that users $i$ can download file $f$ from any other user caching file $f$ at time instant $t$; otherwise $H^f_{i}(t)=0$.
\end{defn}
At time $t$, since user $i$ can download file $f$ when at least one user caching file $f$ is in contact, we get $H^f_i(t)=1- \prod \limits_{j \in \mathcal{S} \backslash \{i\},x_{j,f}=1}\left[1-H_{i,j}(t) \right]$. Similar to \cite{renewalmodel}, it is reasonable to assume that when a user requests a file, the alternating process between each pair of users has been running for a long time. Thus, denote $T^r_{i,f}$, $i \in \mathcal{S}$ and $f \in \mathcal{F}$, as the time of user $i$ requests file $f$, and the communication time $T^c_{i,f}$ can be derived as
$T^c_{i,f}= \lim \limits_{T^r_{i,f} \to \infty} \int _{T^r_{i,f}}^{T^r_{i,f}+T^d} H^f_i(t) dt.$
In the following, we will derive the expectation and variance of the communication time.

\begin{lem} \label{ev}
When user $i \in \mathcal{S}$ requests file $f \in \mathcal{F}$, which is not stored at its own cache, the expectation and variance of its communication time is
\begin{equation} \label{expectation_t}
\mathbb{E}\left[T^c_{i,f} \right]=T^d\left( 1-\prod \limits_{j \in \mathcal{S}, x_{j,f}=1} \frac{\lambda^C_{i,j}}{\lambda^C_{i,j}+\lambda^I_{i,j}} \right).
\end{equation}
and
\begin{align} \label{var_t}
\mathrm{Var} \left[ T^c_{i,f} \right] = & 2 \int_0^{T^d} (T^d-u) \prod \limits_{j \in \mathcal{S}, x_{j,f}=1} \frac{\lambda^C_{i,j}}{(\lambda^I_{i,j}+\lambda^C_{i,j})^2} \notag \\
&\times \left[ \lambda^C_{i,j} + \lambda^I_{i,j} e^{-u(\lambda_{i,j}^C+\lambda^I_{i,j})} \right] du \notag  \\
&-(T^d)^2 \prod \limits_{j \in \mathcal{S}, x_{j,f}=1} \left(\frac{\lambda^C_{i,j}}{\lambda^C_{i,j}+\lambda^I_{i,j}} \right)^2.
\end{align}
\end{lem}

\begin{proof}
See Appendix A.
\end{proof}
Since the communication time $T^c_{i,f}$ is a bounded random variable, we propose to approximate its distribution by a beta distribution, which is widely used to model the random variables limited to finite ranges. Specifically, we consider $T^c_{i,f} \approx T^d Y_{i,f}$, where $Y_{i,f} \sim B(\alpha_{i,f},\beta_{i,f})$, $i \in \mathcal{S}$ and $f \in \mathcal{F}$, if $\sum_{j \in \mathcal{S} \backslash \{i\}} x_{j,f}>0$, which means that user $i$ may download file $f$ from at least one user; otherwise, $T^c_{i,f}=0$. Let $\mathbb{E}[T^dY_{i,f}]=\mathbb{E}[T^c_{i,f}]$ and $\mathrm{Var}[T^dY_{i,f}]=\mathrm{Var}[T^c_{i,f}]$, and the parameters of the beta distribution to match the first two moments can be derived as\footnote{The parameters of the beta distribution should be positive, and it can be proved that $\alpha_{i,f}>0$ and $\beta_{i,f}>0$, by $e^{-u(\lambda_{i,j}^I+\lambda_{i,j}^C)} \le 1$. The detail is omitted due to the space limitation.}
\begin{equation} \label{beta_p}
\begin{cases}
\alpha_{i,f}=\frac{\mathbb{E}[T^c_{i,f}]^2 (T^d-\mathbb{E}[T^c_{i,f}])}{ \mathrm{Var}[T^c_{i,f}]T^d}-\frac{\mathbb{E}[T^c_{i,f}]}{T^d} \\
\beta_{i,f}=\frac{T^d-\mathbb{E}[T^c_{i,f}]}{\mathbb{E}[T^c_{i,f}]} \alpha_{i,f}
\end{cases}
\end{equation}  

\subsection{Data offloading ratio approximation}  
Based on the above approximation, we get an approximate expression of the data offloading ratio in Proposition \ref{E_od}. Simulations will show that the approximation is quite accurate.
\begin{prop} \label{E_od}
The data offloading ratio is approximated as
\begin{align} \label{ratio}
\mathcal{P} = &\frac{1}{N_u}\sum \limits_{i \in \mathcal{S}} \sum \limits_{f \in \mathcal{F}} p^r_{i,f} \left[ x_{i,f} + (1-x_{i,f}) \mathcal{P}_{i,f} \right],
\end{align}
where $\mathcal{P}_{i,f}$ is the data offloading ratio when user $i$ requests file $f$, which is not in its own cache, approximated by
\begin{align} \label{ex_g}
&\mathcal{P}_{i,f} \approx 1-I_{\frac{C}{T^dR}} (\alpha_{i,f},\frac{T^d-\mathbb{E}[T^c_{i,f}]}{\mathbb{E}[T^c_{i,f}]} \alpha_{i,f}) \notag \\
& +\frac{\mathbb{E}[T^c_{i,f}]R}{C} I_{\frac{C}{T^dR}}(\alpha_{i,f}+1,\frac{T^d-\mathbb{E}[T^c_{i,f}]}{\mathbb{E}[T^c_{i,f}]} \alpha_{i,f}) \Big]  \Big\},
\end{align}
if $\sum_{j \in \mathcal{S} \backslash \{i\}} x_{j,f}>0$ and $0$ elsewhere, where $I_r(\cdot,\cdot)$ is the incomplete beta function, and $\alpha_{i,f}$ is given in (\ref{beta_p}).
\end{prop}

\begin{proof}
Following (\ref{define_ratio}),  (\ref{expectation_t}), (\ref{var_t}), and (\ref{beta_p}), the expression in (\ref{ratio}) can be obtained. Due to the space limitation, the detail is omitted.
\end{proof}

\section{Effect of Mobile User Speed}

In this section, we will consider a homogeneous case, where the contact and inter-contact parameters among all pairs of users are the same, i.e., $\lambda^C=\lambda^C_{i,j}>0$ and $\lambda^I=\lambda^I_{i,j}>0$, where $i \in \mathcal{S}$ and $j \in \mathcal{S} \backslash \{i\}$. We will investigate how the user moving speed affects the data offloading ratio for a given caching strategy. If all users cache the same contents, the D2D communications will not help the content delivery. Thus, in the following, we assume that the contents cached at different users are not all the same. This investigation will be based on the approximate expression in (\ref{ratio}), and simulations will be provided later to verify the results.

\subsection{Communication time analysis}
Under the above assumptions, the expectation and variance of the communication time can be simplified, as in the following corollary.
\begin{cor} \label{sim_ev}
When $\lambda^C=\lambda^C_{i,j}$ and $\lambda^I=\lambda^I_{i,j}$, where $i \in \mathcal{S}$ and $j \in \mathcal{S} \backslash \{i\}$, the expectation and variance of a user requests file $f$, which is not stored at its own cache, are given by
\begin{align}
&\mathbb{E}[T^c_{i,f}]=T^d\left[ 1-\left( \frac{\lambda^C}{\lambda^C+\lambda^I} \right)^{n_f} \right],  \label{expect}  \\
&\mathrm{Var}[T^c_{i,f}]= \left[ \frac{\lambda^C}{(\lambda^C+\lambda^I)^2} \right]^{n_f}
\sum \limits_{l=1}^{{n_f}}  \binom{{n_f}}{l} \frac{(\lambda^C)^{n_f-l} (\lambda^I)^{l}}{l(\lambda^C+\lambda^I)} \notag \\
& \quad \times \left[ T^d-\frac{1}{l(\lambda^C+\lambda^I)}+\frac{e^{-l(\lambda^C+\lambda^I)T}}{l(\lambda^C+\lambda^I)} \right], \label{variance}
\end{align}
where $i \in \mathcal{S}$ and $n_f=\sum \limits_{j \in \mathcal{S}} x_{j,f}$ denotes the number of users caching file $f$.
\end{cor}
\begin{proof}
	See Appendix A.
\end{proof}

\subsection{Mobile user speed}

We first characterize the relationship between the user speed and the parameters $\lambda^C$ and $\lambda^I$ in Lemma \ref{speed}.
\begin{lem} \label{speed}
When all the user speeds change by $s$ times, the contact and inter-contact parameters will also change by $s$ times, i.e., from $\lambda^C$ and $\lambda^I$ to $s\lambda^C$ and $s\lambda^I$, respectively.
\end{lem}
\begin{proof}
The time for user $i$ to move along a certain path $L_i$ can be given as a curve integral $\int_{L_i} \frac{dz}{v_i(z)}$, where $v_i(z)$ is the speed of user $i$ when passing by a point $z$ on the path $L_i$. When the speed of user $i$ changes by $s$ times, the time for moving along the path $L_i$ changes to $\int_{L_i} \frac{dz}{s v_i(z)}=\frac{1}{s}\int_{L_i} \frac{dz}{ v_i(z)}$, which scales by $\frac{1}{s}$ times. During each contact or inter-contact time, users $i$ and $j$ move along certain paths. When all the user speeds change by $s$ times, each contact or inter-contact time changes by $\frac{1}{s}$ times, and thus, the average ones change by $\frac{1}{s}$ times. Since the contact and inter-contact times are assumed to be exponential distributed with mean $\frac{1}{\lambda^C}$ and $\frac{1}{\lambda^I}$, respectively, the parameters $\lambda^C$ and $\lambda^I$ scale by $s$ times.
\end{proof}
Considering that a larger $s$ means that users are moving faster, in the following, we will investigate how changing  $s$ will affect the data offloading ratio. For simplicity, we assume that the transmission rate is a constant, and will not change with the user speed. This is reasonable in the low-to-medium mobility regime. Firstly, the effect of user speed on the communication time is shown in Lemma \ref{s_t} .
\begin{lem} \label{s_t}
	When $s$ increases, which is equivalent to increasing the user speed, the expectation of the communication time when a user $i \in \mathcal{S}$ requests file $f \in \mathcal{F}$ that is not in its own cache, i.e., $\mathbb{E}[T^c_{i,f}]$, keeps the same, and the corresponding variance, i.e., $\mathrm{Var}[T^c_{i,f}]$, decreases, if the number of users caching file $f$ is larger than $0$, i.e., $n_f>0$. Accordingly, the parameter $\alpha_{i,f}$ of the beta distribution increases.
\end{lem}
\begin{proof}
	See Appendix B.
\end{proof}
Then, we evaluate the relationship between $\alpha_{i,f}$ and the data offloading ratio when user $i$ requests file $f$ that is not in its own cache, i.e., $\mathcal{P}_{i,f}$ in (\ref{ex_g}), in Lemma \ref{s_g}.
\begin{lem} \label{s_g}
	When user $i \in \mathcal{S}$ requests file $f \in \mathcal{F}$ and cannot find it in its own cache, the data offloading ratio, i.e., $\mathcal{P}_{i,f}$, increases with $\alpha_{i,f}$.
\end{lem}
\begin{proof}
	See Appendix C.
\end{proof}
Base on Lemmas \ref{s_t} and \ref{s_g}, we can specify the effect of user speed on the data offloading ratio in Proposition \ref{s_d}.
\begin{prop}\label{s_d}
	If the transmission rate does not change with the user speed, and the average contact and inter-contact times among all the pairs are the same, the data offloading ratio increases with the user moving speed.
\end{prop}
\begin{proof}
	See Appendix D.
\end{proof}
\begin{remark}
	The result in Proposition \ref{s_d} is valid for any caching strategy, only excluding the case that all the users have the same cache contents.
\end{remark}

\section{Simulation results}
In the simulation, the content request probability follows a Zipf distribution with parameter $\gamma_r$, i.e., $p_f=\frac{f^{-\gamma_r}}{\sum \limits_{i \in \mathcal{F}} i^{-\gamma_r}}$, $f \in \mathcal{F}$ \cite{d2d-cache}. Meanwhile, each user caches 5 contents, and a random caching strategy is applied \cite{randomcache}, where the probabilities of the contents cached at each user are proportional to the file request probabilities.

\begin{figure}[!t]
  \centering
  \includegraphics[width=2.6in]{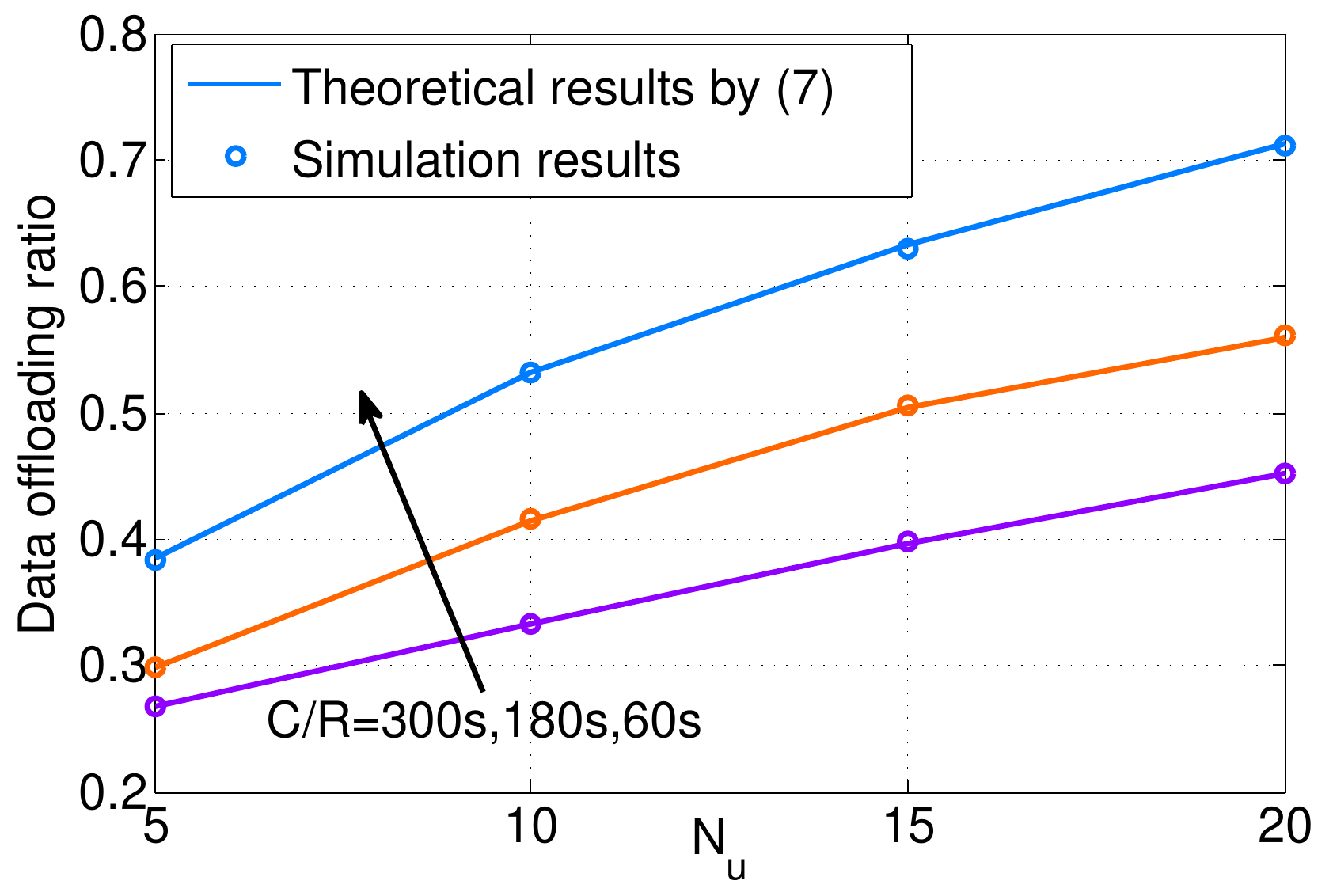}
  \caption{Data offloading ratio with $N_f=100$, $T^d=300s$ and $\gamma_r=0.6$.}
  \label{fig_scheme1}
\end{figure}

Fig. \ref{fig_scheme1} validates the accuracy of the approximation in (\ref{ratio}). The inter-contact parameters $\lambda^I_{i,j}, i \in \mathcal{S}, j \in \mathcal{S} \backslash \{i\}$ are generated according to a gamma distribution as $\Gamma(4.43,1/1088)$ \cite{aggregate_real}. Similar as \cite{renewalmodel}, we assume the average of the contact parameters are $5$ times larger than the inter-contact parameters. Thus, the contact parameters are generated∂ according to $\Gamma(4.43 \times 25,1/1088/5)$. It is shown from Fig. \ref{fig_scheme1} that the theoretical results are very close to the simulation results, which means the approximate expression (\ref{ratio}) is quit accurate. Furthermore, the data offloading ratio increases with the number of users, which is brought by the increasing aggregate caching capacity and the content sharing via D2D links.
\begin{figure}[!t]
  \centering
  \includegraphics[width=2.6in]{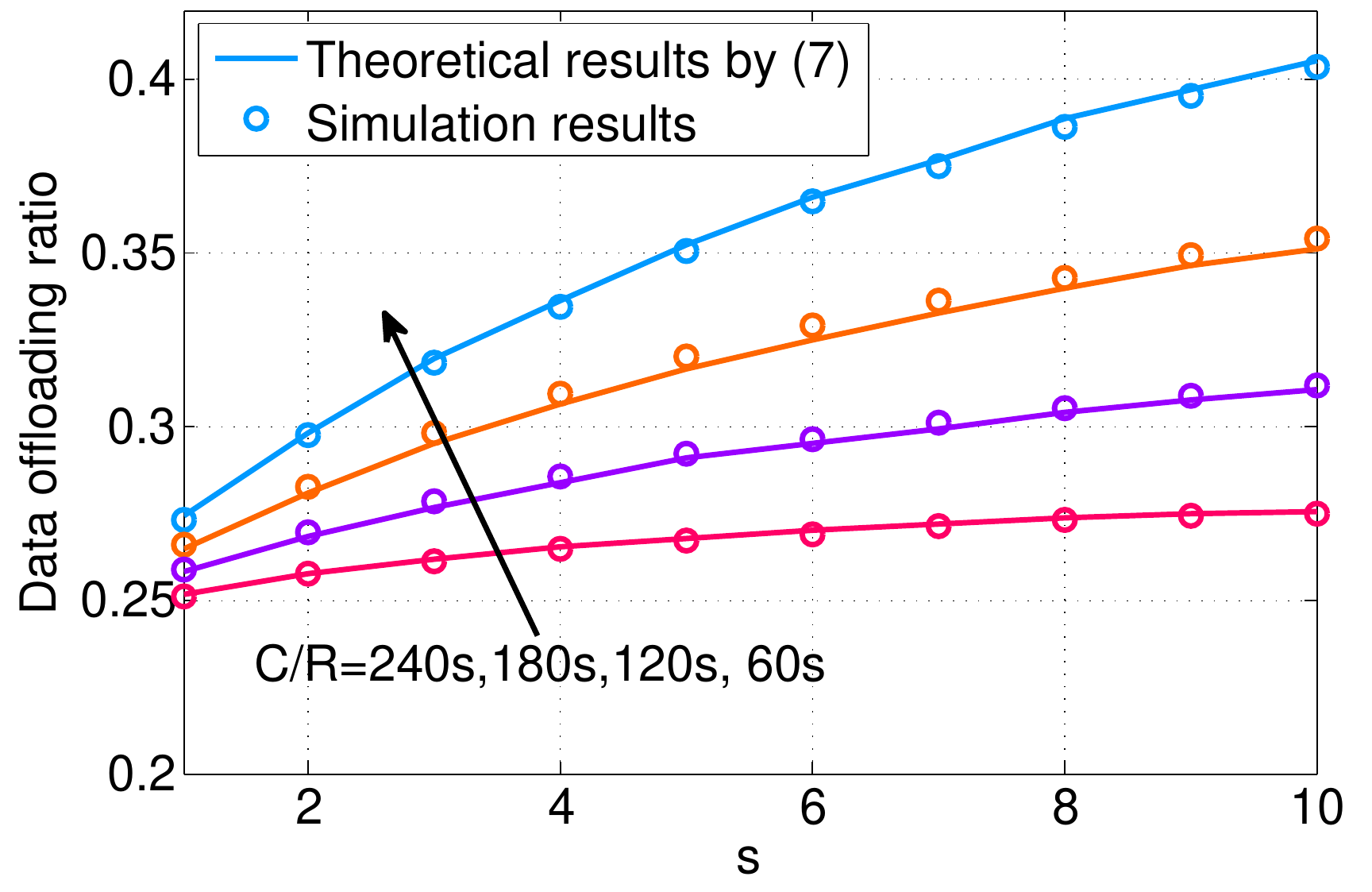}
  \caption{Data offloading ratio with $N_u=15$, $\lambda^C=0.001s$, $\lambda^I=0.0002s$, $N_f=100$, $T^d=300s$ and $\gamma_r=0.6$.}
  \label{fig_scheme2}
\end{figure}

In Fig. \ref{fig_scheme2}, the effect of $s$ is demonstrated, where increasing $s$ is equivalent to increasing the user speed. Firstly, the small gap between the theoretical and simulation results again verifies the accuracy of the approximate expression in (\ref{ratio}). It is also shown that the data offloading ratio increases with $s$, which confirms the conclusion in Proposition \ref{s_d}.  Moreover, from Fig. \ref{fig_scheme2}, the increasing rate of the data offloading ratio is decreases with the user moving speed.


\section{conclusions}

In this paper, we investigated the effect of user mobility on the caching performance in a D2D caching network. The communication time of a given user was firstly approximated by a beta distribution, through matching the first two moments. Then, an approximate expression of the data offloading ratio was derived. For a homogeneous case, where the average contact and inter-contact times are the same for all the user pairs, we evaluated how the user moving speed affects the data offloading ratio. Specifically, it was proved that the data offloading ratio increases with the user speed, assuming that the transmission rate is irrelevant to the user speed. Simulation results validated the accuracy of the approximate expression of the data offloading ratio, and demonstrated that the data offloading ratio increases with the user speed, while the increasing rate decreases with the user speed.

\section*{Appendix}
\subsection{Proof of Lemma \ref{ev} and Corollary \ref{sim_ev}}
As the timeline of different user pairs are independent, the expectation of the communication time when user $i$ requests file $f$, which is not in its own cache, can be written as 
\begin{equation}
\mathbb{E} [ T^c_{i,f} ]=\lim \limits_{T^r_{i,f} \to \infty} \int _{T^r_{i,f}}^{T^r_{i,f}+T^d} \left[ 1- \prod \limits_{j \in \mathcal{S},x_{j,f}=1}\left(1- \mathbb{E} H_{i,j}(t) \right) \right] dt. 
\end{equation}
Since the timeline between each pair of users is modeled as an alternating renewal process, according to \cite{renewalprocess}, we have $\lim \limits_{t \to \infty} \Pr[H_{i,j}(t)=1]=\frac{\lambda^I_{i,j}}{\lambda^C_{i,j}+\lambda^I_{i,j}}$. Thus, $\lim \limits_{t \to \infty} \mathbb{E}[ H_{i,j}(t)]=\frac{\lambda^I_{i,j}}{\lambda^C_{i,j}+\lambda^I_{i,j}}$, and then, the expectation in (\ref{expectation_t}) can be obtained. Let $\lambda^C=\lambda^C_{i,j}$ and $\lambda^I=\lambda^I_{i,j}$, and we can get the expression in (\ref{expect}).
The variance of the communication time is 
\begin{align} \label{v}
&\mathrm{Var} [ T^c_{i,f} ]= \notag \\
&2 \lim \limits_{T^r_{i,f} \to \infty}  \int_{T^r_{i,f}}^{T^r_{i,f}+T^d} \int_{T^r_{i,f}}^{\tau} \Pr[H^f_{i}(t)=1,H^f_{i}(\tau)=1] dtd\tau \notag \\
&-\left( \mathbb{E} [ T^c_{i,f} ] \right)^2
\end{align}
According to \cite{renewalprocess}, $\Pr[H_{i,j}(\tau)=0|H_{i,j}(t)=0]=\frac{\lambda^C_{i,j}}{\lambda^C_{i,j}+\lambda^I_{i,j}}+\frac{\lambda^I_{i,j}}{\lambda^C_{i,j}+\lambda^I_{i,j}}e^{-({\lambda^C_{i,j}+\lambda^I_{i,j}})(\tau-t)}$. Then, when $T_{i,f}^r \to \infty$, we can get
\begin{align} \label{prob}
&\Pr[H^f_{i}(\tau)=1,H^f_{i}(t)=1] 
=1-2 \prod \limits_{j \in \mathcal{S},x_{j,f}=1}\frac{\lambda^C_{i,j}}{\lambda^C_{i,j}+\lambda^I_{i,j}} \notag \\
& \quad + \prod \limits_{j \in \mathcal{S},x_{j,f}=1} \frac{\lambda^C_{i,j}}{(\lambda^I_{i,j}+\lambda^C_{i,j})^2}\left[ \lambda^C_{i,j} + \lambda^I_{i,j} e^{-(\lambda_{i,j}^C+\lambda^I_{i,j})(\tau-t)} \right]
\end{align}
Let $u=\tau-t$ and substitute (\ref{prob}) into (\ref{v}), and we can get (\ref{var_t}). Let  $\lambda^C=\lambda^C_{i,j}$ and $\lambda^I=\lambda^I_{i,j}$, and we can get (\ref{variance}) with the binomial theorem.

\subsection{Proof of Lemma \ref{s_t}}

When the user speed changes by $s$ times, the expectation of the communication time in (\ref{expect}) keeps the same, while the variance changes to
\begin{align}
&\mathrm{Var}[T^c_{i,f}]= \left[ \frac{\lambda^C}{(\lambda^C+\lambda^I)^2} \right]^{n_f}
\sum \limits_{l=1}^{{n_f}}  \binom{{n_f}}{l} \frac{(\lambda^C)^{n_f-l} (\lambda^I)^{l}}{sl(\lambda^C+\lambda^I)} \notag \\
& \quad \times \left[ T^d-\frac{1}{sl(\lambda^C+\lambda^I)}+\frac{e^{-sl(\lambda^C+\lambda^I)T^d}}{sl(\lambda^C+\lambda^I)} \right],
\end{align}
To prove that $\mathrm{Var}[T^c_{i,f}]$ decreases with $s$, we will prove that $\frac{
\partial \mathrm{Var}[T^c_{i,f}]}{\partial s}<0$. The partial derivation of $\mathrm{Var}[T^c_{i,f}]$ is
\begin{align} \label{dir}
&\frac{\partial \mathrm{Var}[T^c_{i,f}]}{\partial s}= \notag \\
&\left[ \frac{\lambda^C}{(\lambda^C+\lambda^I)^2} \right]^{n_f}
\sum \limits_{l=1}^{{n_f}}  \binom{{n_f}}{l} \frac{(\lambda^C)^{n_f-l} (\lambda^I)^{l}}{s^3l^2(\lambda^C+\lambda^I)^2} \mathcal{A}_1(x),
\end{align}
where
$\mathcal{A}_1(x)=-x-x e^{-x}-2(e^{-x}-1)$ and $x=sl(\lambda^C+\lambda^I)T^d>0$.
Since $\mathcal{A}'_1(x)=-1+(1+x) e^{-x} < -1+(1+x) \frac{1}{1+x}=0$, $\mathcal{A}_1(x)$ is a decreasing function of $x$. Thus, $\mathcal{A}_1(x)<\mathcal{A}_1(0)=0$. According to (\ref{dir}), when $n_f>0$, we have $\frac{
\partial \mathrm{Var}[T^c_{i,f}]}{\partial s}<0$. The parameter $\alpha_{i,f}$ given in (\ref{beta_p}) is a decreasing function of $\mathrm{Var}[T^c_{i,f}]$, and thus increases with $s$.

\subsection{Proof of Lemma \ref{s_g}}
To simplify the expression in (\ref{ex_g}), denote $r \triangleq \frac{C}{T^d R} \in (0,1)$, $y \triangleq \frac{T^d-\mathbb{E}[T^c_{i,f}]}{\mathbb{E}[T^c_{i,f}]} \ge 0$, and $\alpha \triangleq \alpha_{i,f}$. The expression in (\ref{ex_g}) can be rewritten as a function of $\alpha$, given as
\begin{align}
	&\mathcal{P}_{i,f}=1- \frac{\int_0^r (1-\frac{u}{r}) u^{\alpha-1} (1-u)^{y \alpha-1} du }{B(\alpha,y\alpha)}.
\end{align}
Let $g(\alpha)=1-\mathcal{P}_{i,f}$, the derivation of $g(\alpha)$ is
\begin{align}
	&g'(\alpha)= \notag \\
	&\frac{1}{B(\alpha,y\alpha)} \Bigg\{\int_0^r (1-\frac{u}{r}) u^{\alpha-1} (1-u)^{y \alpha-1}[\ln u + y \ln (1-u)] du \notag \\
	&- \int_0^r (1-\frac{u}{r}) u^{\alpha-1} (1-u)^{y \alpha-1} du D(y,\alpha) \Bigg\},
\end{align}
where $D(y,\alpha)=\psi(\alpha)+y \psi(y\alpha)-(1+y)\psi[(1+y)\alpha]$ and $\psi(\cdot)$ is the digamma function. If $r=1$, $g'(\alpha)=\frac{\partial [y/(1+y)]}{\partial \alpha}=0$.
Denote $\mathcal{A}_2(r)=\frac{B(\alpha,y\alpha)}{r}g'(\alpha)$, $\mathcal{A}_2(1)=0$ and
\begin{align} 
	&\lim \limits_{r \to 0^{+}} \mathcal{A}_2(r)= \notag \\
&\lim \limits_{r \to 0^{+}} \int_0^r (r-u) u^{\alpha-1} (1-u)^{y \alpha-1}[\ln u + y \ln (1-u)] du \notag \\
\end{align}
Since $ r \ge u \ge 0$ and $y \ge 0$, $(r-u) u^{\alpha-1} (1-u)^{y \alpha-1} \ge 0$ and $\ln u + y \ln (1-u) \le 0$, thus, $\lim \limits_{r \to 0^{+}} \mathcal{A}_2(r) \le 0$. The derivation of $\mathcal{A}_2(r)$ is 
\begin{align}
	\mathcal{A}'_2(r)= &\int_0^r u^{\alpha-1} (1-u)^{y \alpha-1}[\ln u + y \ln (1-u)] du \notag \\
	&- \int_0^r u^{\alpha-1} (1-u)^{y \alpha-1} du D(y,\alpha).
\end{align}
Thus, $\mathcal{A}'_2(1)=\frac{\partial B(\alpha,y \alpha)}{\partial \alpha}-\frac{\partial B(\alpha,y \alpha)}{\partial \alpha}=0$ and $\lim \limits_{r \to 0^{+}} \mathcal{A}'_2(r) \le 0$.
Then, we can get $\mathcal{A}''_2(r)= r^{\alpha-1} (1-r)^{y \alpha-1}[\ln r + y \ln (1-r)-D(y,\alpha)]$. Let $\mathcal{A}_3(r)=r^{1-\alpha} (1-r)^{1-y \alpha} \mathcal{A}''_2(r)$, then, there is one zero point of $\mathcal{A}'_3(r)=\frac{1-(1+y)x}{x(1-x)}$ in $(0,1]$. Thus, there is one inflection point of $\mathcal{A}_3(r)$. Considering that $\lim \limits_{r \to 0^{+}}\mathcal{A}_3(r)=\lim \limits_{r \to 1^{-}}\mathcal{A}_3(r)=-\infty$, the sign of $\mathcal{A}_3(r)$ may be negative, or first negative, then positive, and then negative, while $r$ increases in $(0,1)$. If $\mathcal{A}_3(r)<0$, then $\mathcal{A}''_2(r)<0$ when $r \in (0,1)$. However, we have $\lim \limits_{r \to 0^{+}} \mathcal{A}'_2(r) \le \mathcal{A}'_2(1)$, which means that $\mathcal{A}'_2(r)$ can not be a decreasing function in $(0,1)$. Thus, the sign of $\mathcal{A}_3(r)$ is first negative, then positive, and then negative, while $r$ increases in $(0,1)$. Since $\mathcal{A}''_2(r)$ has the same sign with $\mathcal{A}_3(r)$ in $(0,1)$, $\mathcal{A}'_2(r)$ first decreases, then increases, and then decreases while $r$ increases in $(0,1)$. Considering that $\lim \limits_{r \to 0^{+}} \mathcal{A}'_2(r) \le 0$ and $\mathcal{A}'_2(1)=0$, the sign of $\mathcal{A}'_2(r)$ must be first negative, and then positive in $(0,1)$. Therefore, while $r$ increases in $(0,1)$, $\mathcal{A}_2(r)$ first decreases, and then increases. Considering that $\lim \limits_{r \to 0^{+}} \mathcal{A}_2(r) \le 0$ and $\mathcal{A}_2(1)=0$, we have $\mathcal{A}_2(r)<0$ in $(0,1)$ and $\mathcal{A}_2(r)=0$ when $r=1$. Since $g'(\alpha)=\frac{r}{B(\alpha,y\alpha)} \mathcal{A}_2(r)$, we get $g'(\alpha)<0$ in $(0,1)$. Thus, $g(\alpha)$ decreases with $\alpha$, and $\mathcal{P}_{i,f}=1-g(\alpha)$ increases with $\alpha$.

\subsection{Proof of Proposition \ref{s_d}}
The data offloading ratio in (\ref{ratio}) increases with the increasing of $\mathcal{P}_{i,f}$ if $x_{i,f}=0$, $i \in \mathcal{S}$, $f \in \mathcal{F}$. Then, based on Lemmas  \ref{s_t} and \ref{s_g}, we can get that the data offloading ratio when user $i$ requests file $f$ from other users, i.e., $\mathcal{P}_{i,f}$, decreases with the user speed when $n_f>0$, otherwise $\mathcal{P}_{i,f}=0$. Accordingly, the data offloading ratio when user $i$ requests file $f$, i.e., $x_{i,f}+ (1-x_{i,f}) \mathcal{P}_{i,f} $, increases with the user speed when $x_{i,f}=0$ and $n_f>0$; otherwise, it keeps the same, where $i \in \mathcal{S}$, $f \in \mathcal{F}$. Since we consider that not all the users cache the same contents, there must exists $i' \in \mathcal{S}$, $j' \in \mathcal{S}$ and $f' \in \mathcal{F}$, where $x_{i',f'}=0$ and $x_{j',f'}=1$, i.e, $n_{f'}>0$. Thus, the data offloading ratio increases with the user speed.

\bibliographystyle{IEEEtran}
\bibliography{IEEEabrv,report}
\end{document}